\documentclass[runningheads]{llncs}

\usepackage{graphicx}
\usepackage[caption = false]{subfig}
\usepackage{amssymb}
\usepackage{amstext}
\usepackage{fixltx2e}
\usepackage{letltxmacro}
\usepackage{pifont}
\usepackage{multirow}
\usepackage{array}
\usepackage{color}
\usepackage{algorithm}
\usepackage{algorithmic}
\usepackage{nth}

\newcommand{\soudeh}[1]{{\color{blue}{}}} 
\newcommand{\fixme}[1]{{\color{red}{}}} 

\newcolumntype{P}[1]{>{\centering\arraybackslash}p{#1}}
\newcolumntype{M}[1]{>{\centering\arraybackslash}m{#1}}

\newtheorem{assumption}[theorem]{Assumption}

\begin{document}

\title{KPsec: Secure End-to-End Communications for Multi-Hop Wireless Networks}

\author{Mohammed Gharib\inst{1}
\and
Ali Owfi\inst{2}
\and
Soudeh Ghorbani\inst{1}
}

\authorrunning{M. Gharib et. al.}

\institute{Johns Hopkins University, Baltimore, MD 21218, USA \\
\email{gharib@jhu.edu, soudeh@cs.jhu.edu} \\\and
Sharif University of Technology, Tehran, Iran\\
\email{owfi@ce.sharif.edu}}

\maketitle     

\begin{abstract}

The security of cyber-physical systems, from self-driving cars to medical devices, depends on their underlying multi-hop wireless networks. Yet, the lack of trusted central infrastructures and limited nodes' resources make securing these networks challenging. Recent works on key pre-distribution schemes, where nodes communicate over encrypted overlay paths, provide an appealing solution because of their distributed, computationally light-weight nature. Alas, these schemes share a glaring security vulnerability: the two ends of every overlay link can decrypt---and potentially modify and alter---the message. Plus, the longer overlay paths impose traffic overhead and increase latency.

We present a novel routing mechanism, KPsec, to address these issues. KPsec deploys multiple disjoint paths and an initial key-exchange phase to secure end-to-end communications. After the initial key-exchange phase, traffic in KPsec follows the shortest paths and, in contrast to key pre-distribution schemes, intermediate nodes cannot decrypt it. We measure the security and performance of KPsec as well as three state-of-the-art key pre-distribution schemes using a real 10-node testbed and large-scale simulations. Our experiments show that, in addition to its security benefits, KPsec results in $5-15\%$ improvement in network throughput, up to $75\%$ reduction in latency, and an order of magnitude reduction in energy consumption. 
 
\end{abstract}

\section{Introduction}
\label{introduction}
Cyber physical systems (CPS) are increasingly deployed in mission-critical systems such as self-driving cars \cite{CPS1}. While most of such systems could be implemented with expensive infrastructure, the better solution is to implement them based on the peer-to-peer network node cooperation \cite{CPS2}. Smart intersections, where the cars never stop at a red light unless there will be actual crossing traffic, is an instance \cite{smartIntersect}. Vehicle to vehicle communication is another example while it can potentially help to prevent $74\%$ of all traffic accidents including those with drivers impaired by alcohol or drowsiness, as reported by national highway traffic safety administration of the U.S. \cite{NHTSA}. 
To be able to rely on these systems, the security of the underlying multi-hop wireless networks, such as mobile ad hoc networks (MANET), vehicular ad hoc networks (VANET), and wireless sensor networks (WSN), is critical. Alas, the lack of trusted infrastructures and limited node resources make securing communications in such networks challenging. Concretely, while cryptography is a general and powerful approach to improve security, it is not well suited for such networks. This is because cryptography techniques, such as public key infrastructure (PKI), commonly rely on a key management system and most of the key management tasks are assigned to a trusted third party (TTP) or several distributed TTPs that are based on infrastructure. In contrast, multi-hop wireless networks in cyber-physical systems are fully decentralized and lack a fixed infrastructure that can act as the TTP. Plus, nodes in such networks have limited memory, computational, and transmission resources. Consequently, the naive solution of storing all keys in every single node for encrypting and decrypting messages is also not practical in these networks, especially in large-scale ones.

Key pre-distribution schemes \cite{gligor} seem to be a promising solution due to their distributed and lightweight nature. Key pre-distribution schemes store just $k$ keys in each node, where $k<<n$ and $n$ is the number of network nodes. The set of stored keys in each node is referred to as its \emph{keyring}. Once a node encrypts a message with a key, only those nodes with a shared key are capable of decrypting it. Thus, a pair of nodes can communicate directly and securely if they share a common key. To establish a secure connection between two nodes without a shared key, a \emph{key-path} has to be found. The key-path is an overlay path in which each pair of adjacent nodes have a \emph{secure link} between them\footnote{Note that this secure overlay link may span multiple physical nodes,  in reality.}, i.e., they share a common key. To exchange messages, the source initially encrypts its message and forwards it to the first hop on the overlay. The message is then routed over the overlay where each intermediate hop, in turn, decrypts the data, encrypts it again with a key shared with the next hop, and forwards it to the next hop toward the destination.

Despite important differences between various classes of key pre-distribution techniques (such as symmetric and asymmetric cryptosystems) in terms of their routing mechanisms and the process of forming secure overlays (\S\ref{sec:relatedWork}), they fundamentally share a security vulnerability, known as \emph{intermediate D-E steps} or \emph{per hop key exposure} \cite{D-E} where the intermediate nodes on the key-path overlay can decrypt and encrypt messages. Since an attacker can compromise an intermediate node, any decryption-encryption (D-E) step raises a security threat. While enhancing the link-level security of the key pre-distribution schemes has been the focus of many recent works \cite{qcomp,ning,linkSecurity,linkSecurity2,linkSecurity3,linkSecurity4}, the holistic, end-to-end security of these schemes is relatively unexplored. In addition to this security concern, the performance of key pre-distribution schemes is not ideal because their overlay paths are commonly longer than the physical shortest paths. The resulted path stretch leads to performance degradation, e.g., increased latency and network overhead, as we quantify in \S\ref{sec:experimental}.

In this paper, we propose Key Pre-distribution security (KPsec), a high-performance algorithm to establish end-to-end secure communications in multi-hop wireless networks. Under KPsec, the source and the destination first engage in an initial phase of exchanging public keys via multiple disjoint paths. KPsec leverages a state-of-the-art asymmetric key pre-distribution technique, probabilistic asymmetric key pre-distribution (PAKP) \cite{waina}, as a building block to initially exchange public keys. This step is followed by constructing shared keys for this communicating pair before they start secure communication over the shortest paths. Despite the initialization cost and delay, we show that the amortized latency overhead is low in our scheme compared to the state-of-the-art. This is because upon constructing a common key, under KPsec, traffic follows the shortest path, in lieu of the longer overlays deployed in key pre-distribution techniques. KPsec is not subject to passive attacks due to exchanging only public keys. Moreover, we show that it has high resiliency against active attacks (\S\ref{sec:experimental}).

Concretely, the core idea of KPsec is an initial key exchange process that results in a pairwise key agreement\footnote{We show in \S\ref{sec:experimental} the energy efficiency enabled by this approach.} between the source and destination. After the completion of the key-exchange phase, messages between these two nodes will be encrypted using this key and can be decrypted only by the source and the destination. Thus, these messages can be routed over the shortest physical paths, avoiding the longer key-path overlays without compromising security. Applied naively, this technique is prone to man-in-the-middle attacks where an intermediate node that participates in the key-exchange process replaces the actual key with its own key. In such a case, the intermediate node can read and potentially alter the message while remaining hidden from the source and the destination. To increase the resiliency against this type of attack, nodes in KPsec exchange keys via multiple, and disjoint, paths using erasure coding (\S\ref{sec:proposedAlgorithm}).

Despite its security and performance benefits, KPsec causes a few concerns. First, while path redundancy improves security, the communication can still be vulnerable to more sophisticated forms of attack such as distributed, coordinated man-in-the-middle attacks where a group of compromised nodes agrees on a forged key to replace the actual key. We experimentally show that KPsec has strong resiliency against this type of attack: the attacker needs to compromise $O(n)$ nodes to be able to get access to the secret data.
Second, for the proposed algorithm to work, there must be enough reasonably short overlay vertex disjoint paths for the initial step of exchanging keys. In Section (\ref{sec:proposedAlgorithm}), we investigate the expected number and lengths of these paths. Our results show that, although a large number of disjoint paths improves security, KPsec results in high degrees of security even with small number of such paths, e.g., for relatively short paths with 3 D-E steps, KPsec's use of only 5 disjoint paths leads to 99.9\% resiliency against node capture (\S\ref{sec:proposedAlgorithm}). Third, the initial key-exchange phase causes some control overhead. Our measurements show that the amortized traffic overhead is low. This is because once the key-exchange phase terminates, traffic follows the shortest paths, eliminating the path stretch and compensating for the commencing control overhead. For a network with 100 nodes, for example, KPsec results in almost equal control traffic compared to three state-of-the-art key pre-distribution schemes that we use as baselines and 7.5\% enhancement in throughput (\S\ref{sec:experimental}).

To comprehensively evaluate the performance and security of KPsec, we implement it on a 10-node testbed and a large-scale ns-2 network simulator \cite{ns2}. In addition to KPsec, we implement three state-of-the-art key pre-distribution schemes: PAKP \cite{waina}, unital key pre-distribution (UKP) \cite{TWC2013}, and strong Steiner trade (SST) \cite{infocom2011}. Moreover, to make the end-to-end connections in UKP and SST secure, we augment these algorithms using the design presented in \cite{globcom05} (hereafter called augmented UKP and SST), a general remedy for the intermediate D-E steps problem which is applicable to any symmetric key pre-distribution scheme. Our experiments show that, compared to these baselines, KPsec results in $5-15\%$ throughput improvement, reduces the network latency by $50-75\%$, and alleviates the energy consumption up to an order of magnitude.

Although the performance of KPsec and augmented UKP and SST are close, KPsec results in substantial security improvements, as it is the only scheme that is secure against passive attacks. Plus, an active attacker needs to compromise $O(n)$ nodes to access data in KPsec, a substantially larger fraction compared to augmented UKP and SST. Moreover, contrary to other schemes that suffer from the secret information leakage, compromising a few nodes in KPsec does not enable an attacker to access any secret information. Finally, while in other schemes, sophisticated attackers such as those carrying out selective node compromise attacks can compromise the entire network communications by capturing only a few nodes, in KPsec, compromising the entire network requires the attacker to capture $O(n)$ nodes, e.g., in a network with 100 nodes, to access data, the attacker needs to capture, respectively, 10 and 23 nodes in augmented UKP and SST, compared to 99 nodes in KPsec. 

The main contribution of this paper is KPsec, an algorithm to establish end-to-end secure communication in multi-hop wireless networks, and its thorough evaluation. More specifically, this paper proposes an algorithm to address the two key pre-distribution shortcomings, intermediate D-E steps, and the path stretch and studies its security and performance compared to the state-of-the-art algorithms using a real 10-node testbed and large-scale simulations.

\section{Related Work}
\label{sec:relatedWork}

Key pre-distribution schemes are categorized into two main categories based on their underlying cryptosystem, symmetric and asymmetric. In this section, we provide a brief comparison of these categories, the state-of-the-art techniques for each, and proposals for secure end-to-end communications using key pre-distribution in turn.

\subsection{Symmetric vs. Asymmetric Key Pre-distribution} 
The core idea of symmetric key pre-distribution schemes, which is also known as pairwise key pre-distribution, was first introduced by Eschenauer and Gligor \cite{gligor}. In this scheme, each keyring is chosen uniformly at random from a key-pool, with replacement. The main security shortcoming of the Eschenauer-Gligor design is that if an attacker compromises several nodes, it can access many keys from the key-pool. Thus, many links inside the network become insecure. Chan et al. \cite{qcomp} propose Q-composite algorithm to mitigate this security shortcoming by establishing secure links only between nodes that have at least $q$ common keys.

More recently, the concept of combinatorial block design is used in \cite{TWC2013,infocom2011} to build key pre-distribution schemes. 
Bechkit et al. \cite{TWC2013} propose a key pre-distribution scheme based on unital block design, referred to as naive unital key pre-distribution (NU-KP). The proposed scheme has a low key-sharing probability: $O(\frac{1}{k})$. To improve this probability, they suggest to pre-load each node with $t$ disjoint blocks and refer to the new design as t-UKP. Ruj et al. \cite{infocom2011} propose a method to construct strong Steiner trade (SST), a form of block design, and use it as a key pre-distribution scheme. SST establishes a unique secret pairwise key between nodes. It is proven that the probability of sharing such a pairwise key does not exceed $0.25$ \cite{TWC2013}. In our evaluations, we implement 2-UKP and SST as two well-known baseline schemes.

Liu et al. \cite{AKPS} introduce the idea of asymmetric key pre-distribution, relying on some keying material servers. Multi-hop wireless networks, however, do not always have access to keying servers. Probabilistic asymmetric key pre-distribution (PAKP) was subsequently proposed to consider this problem \cite{waina}. In this scheme, each node stores $k$ public keys chosen uniformly at random with replacement, from a key-pool containing all the public keys. In \cite{waina}, authors prove that in PAKP, for any $k\ge 3$, the probability of key-path existence is more than $99.9\%$, where the impact of increasing the number of nodes is negligible. They further prove that PAKP reduces the average number of D-E steps to $O(\log_k n)$. In comparison, this number is in the order of the physical path length in symmetric key pre-distribution schemes. In contrast to the random key distribution, authors of \cite{SIN} and \cite{AINA} propose and analyze several more realistic scenarios for asymmetric key distribution. 

While the general paradigm of key pre-distribution is similar for both categories---symmetric and asymmetric cryptosystems---the routing policies of these two categories have significant differences. In symmetric systems, a key-pool containing all the secret keys is formed. Any node is pre-loaded with a keyring chosen from the key-pool. During a shared key discovery process, any two adjacent nodes discover their secure link by checking whether they share a common key or not. Accordingly, to find a secure path from the source node to the destination, a physical path is first found. Subsequently, for any physical hop, if there is no secure link, a key-path is found. The transferred data is then encrypted by the source node, decrypted and encrypted again by each intermediate node until reaching the destination. 

In asymmetric key pre-distribution, on the other hand, the routing mechanism follows a reverse process: the key-pool is formed by the public keys of all nodes. Each node is pre-loaded by $k$ public keys chosen uniformly at random with replacement from the pool. Initially, a key-path from the source node to the destination has to be found. Subsequently, for any key-path hop, the corresponding physical path is selected. In this case, each key-path hop may contain several physical hops without decryption and encryption steps, since the overlay neighbors may be physically far away. Generally, there are three main differences between symmetric and asymmetric key pre-distribution schemes. First, the routing process follows a reverse routing procedure. Second, the distributed keys are not confidential. Third, the overlay links in asymmetric schemes are directed. 

In a key pre-distribution scheme, regardless of the symmetric or asymmetric nature of its relaying cryptosystem, there are some intermediate nodes which decrypt the data, encrypt it again, and forward it toward the destination. Since the adversary node may forge itself as an intermediate node, any D-E step is considered as a security threat. Moreover, the resulted path may also be longer than the shortest physical path, due to the absence of a direct secure link, which leads to performance degradation. 
We provide more details about these two categories of key pre-distribution techniques below.

\subsection{End-to-End Communication}
While the intermediate D-E steps problem was first introduced in \cite{D-E}, this work does not propose a solution. To the best of our knowledge, the algorithm of \cite{globcom05} is the first well-defined end-to-end solution for intermediate D-E steps. In this solution, the source node chooses a pairwise key, splits it into $\rho$ pieces, and sends each piece via different node-disjoint paths to the destination. In this way, the attacker needs to compromise at least one node from each node-disjoint path to retrieve the entire pairwise key and decrypt the data. To improve the performance of \cite{globcom05}, Li et al. \cite{co-next05} suggest using intermediate nodes as proxies, and then use multiple paths, each path with just one proxy, to send the key pieces. Gupta et al. \cite{LNCS06} propose their algorithm based on \cite{co-next05} by introducing some proxies as friends. They then use a publicly known function and only the key pieces of the friends to retrieve the pairwise key. Sheu et al. \cite{comcom07} propose using a group-based pairwise key to enhance the security of node-disjoint paths. However, this algorithm requires a group-key agreement. A security shortcoming shared across all these algorithms is their reliance on sending secret values (e.g., private keys) through hop-by-hop D-E steps to establish a pairwise key. The attacker will be able to access these values, and consequently encrypted messages, via compromising the intermediate nodes. Similar to KPsec, \cite{yair} strives to establish end-to-end secure communications by providing disjoint overlay paths. Unlike KPsec, however, it relies on a backbone infrastructure.
\section{KPsec: end-to-end secure communications}
\label{sec:proposedAlgorithm}
KPsec is, in essence, a three-phase algorithm---the source and the destination initially engage in a public key exchange process to build a common key (phases 1 and 2). Their messages are then encoded using this key and routed, securely and efficiently, over shortest paths (phase 3). After presenting an overview of the algorithm, we analyze its key aspects such as the number and the length of disjoint key-paths and its resilience against cooperative attacks in turn.

\subsection{The Three Phases of KPsec}
In the first phase, the goal of the source is to send its public key to the destination efficiently and securely. For this, KPsec leverages multiple vertex-disjoint paths and the notion of erasure coding. Erasure code, a method originally developed for forward error correction code under bit erasures, transforms a message into a longer coded message with redundant data pieces. This coded message is then broken into $\rho$ shares such that the original message can be recovered from any $\theta$ shares. After encoding and splitting its public key, KPsec then sends the $\rho$ shares to the destination over vertex disjoint paths. Splitting the key into $\rho$ shares and sending them over disjoint paths make the system more resilient against the man-in-the-middle attack---the attacker needs to compromise more nodes to be able to forge the public key.

In the second phase, the destination node collects the shares and extracts the public key of the source. It then encrypts its own public key using the public key of the source and sends this encrypted message via the shortest physical path toward the source.

In the third and final phase, both the source and the destination calculate a pairwise key. The source node then encrypts its data using the pairwise key and sends it toward the destination. The destination, in turn, decrypts the data using the same pairwise key. Although we could use any asymmetric cryptography algorithm, we deployed elliptic curve cryptography (ECC) \cite{ecc} because of its shorter key length and lower computational complexity compared to other asymmetric cryptography algorithms. 
In the rest of this section, after outlining our assumptions, we describe the details of each phase. Table (\ref{tableOfNotatio}) lists the notations that we will use throughout this paper.

\begin{assumption}
The asymmetric cryptosystem security strength is such that, by having the public key and other public parameters, the attacker is unable to compute the private key.
\end{assumption}
\begin{assumption}
When there is more than one path toward the destination and the source node randomly chooses one of them, the attacker cannot guess which path is chosen. 
\end{assumption}
\begin{table}[t]
\caption[Table of notations]{The table of notations. }
\begin{minipage}{\textwidth}
\begin{tabular}{ l | l }
$n$ & Number of network nodes \\
$k$ & Size of key-ring \\
$\rho$ & Sufficient number of vertex-disjoint paths\\
$\theta$ & Threshold for the number of duplicated keys\\ 
$P(i)$ & Erasure code polynomial in the order of $\theta$\\
$\mathcal{P}$ & Set of collected shares by the destination \\
$x_{i}$ & Private key of node $i$\\
$y_{i}$ & Public key of node $i$\\
$k_{sd}$ & The source-destination pairwise key\\
\end{tabular}
\label{tableOfNotatio}
\end{minipage} 
\end{table} 

\textbf{Phase 1:}
The source node chooses $\theta$ random numbers $a_1,a_2,\dots,a_\theta$ and forms the following formula: 
\begin{equation}
\label{polynomial}
P(x)=y_{src}+a_1 x+a_2 x^2+\dots+a_\theta x^\theta.
\end{equation}
This polynomial is used for coding where $P(0)$ is the public key of the source node. The source node calculates $P(i)$, $i=1,2,\dots,\rho$ and then calculates $sign(P(i))$ which is the value of $P(i)$ signed by the private key of the source node. It could be used to certify the correctness of the shares. The source node then sends each tuple $(i , P(i) , sign(P(i)))$ from the $i^\text{th}$ vertex-disjoint overlay path.

\textbf{Phase 2:}
In this phase, the destination collects $\theta$ shares and forms the set $\mathcal{P}$. It then calculates the public key of the source node as
\begin{equation}
P(0)=\sum_{i\in\mathcal{P}}^{} P(i) l_i,
\end{equation}
where $l_i$ is Lagrange multiplier and could be calculated as \begin{equation}
l_i=\prod_{j\in\mathcal{P}, j\neq i}^{}\frac{(0-j)}{(i-j)}.
\end{equation}
Note that the computational complexity of the mentioned erasure code is $O(\theta^2)$.
If $\theta=\rho$, the Lagrange multipliers become unique and thus each node can simply store them. In this case, the computational complexity of the code reduces to $O(\theta)$. By calculating the public key of the source node, the destination node can certify the shares by checking the sign of each share. The destination then encrypts its own public key with the public key of the source node and sends it through the shortest physical path. The source node decrypts the destination's public key. At this point, both ends have exchanged their public keys. 

\textbf{Phase 3:}
In principle, the source is now able to communicate with the destination directly and securely, using asymmetric encryption.
However, asymmetric encryption is known to be computationally complex and energy inefficient. Therefore, it is not an ideal choice for multi-hop wireless networks. KPsec uses symmetric encryption instead: upon receiving each other's public key, source and destination nodes calculate a pairwise key $k_{sd}$:
\begin{equation}
k_{sd}=x_s.y_d=x_d.y_s.
\end{equation}
Since in ECC the corresponding public key for the private key $x$ is calculated as $y=x. G$ where $G$ is the elliptic cure base point, the pairwise key $k_{sd}$ will be identical for both the source and the destination nodes. After this step, the source node can encrypt its data using the pairwise key and then sends it to the destination via the shortest physical path. The destination can also use the same key to decrypt the received data.

KPsec raises a few concerns. Specifically, the first phase of the algorithm relies on a number of disjoint paths. Its operation, security, and performance, therefore, hinges on the existence and lengths of such paths. Moreover, the resilience of the algorithm against cooperative attacks, where the attacker controls a fraction of all nodes, is not known. In the rest of this section, we perform a comprehensive analysis to address these concerns.

\subsection{Number and length of Vertex-Disjoint Key-Paths}
\label{sec::keyPathNo}
Before calculating the number and the length of vertex-disjoint overlay paths in the KPsec algorithm, we need to know how many vertex-disjoint paths KPsec requires. Equivalently, what is the proper value for parameter $\rho$? Furthermore, we need to know how we can find a set of vertex-disjoint paths. To answer the first question, we use the reliability analysis technique of \cite{trivedi}, referred to as the reliability of the series-parallel systems. In this technique, the reliability of the system is considered as the probability of system success which is equal to one minus the probability of attacker success. 
\begin{lemma}
Consider the probability of each intermediate node to be compromised as $p$, the reliability of multi-path systems is equal to
\begin{equation}
\label{eq::reliability}
R=1-(1-(1-p)^{\#DE})^{\rho},
\end{equation}
where $\#DE$ represents the number of intermediate D-E steps in each path.
\end{lemma}
\begin{proof}
In this analysis, each intermediate D-E step is considered as a reliability threat.
Hence, for a path to be reliable, it should be empty of any compromised node. Hence, the reliability of each path is equal to $(1-p)^{\#DE}$. The attacker has to compromise at least an intermediate node from each path to potentially becomes able to perform a successful attack. The probability of attacker success in each path is one minus the reliability of that path, i.e. $1-(1-p)^{\#DE}$. For $\rho$ disjoint paths, hence, the total reliability is equal to 
\begin{equation}
R=1-(1-(1-p)^{\#DE})^{\rho}.
\nonumber
\end{equation}
\end{proof} 
Fig. (\ref{Fig::reliability}) shows the quantitative results of Equation (\ref{eq::reliability}) for a network with $10\%$ of nodes being compromised, selected uniformly at random, i.e. $p=0.1$. While increasing the number of paths improves security, Fig. (\ref{Fig::reliability}) shows that after the first several paths, the security improvement of adding extra paths is negligible. 

\begin{figure}[t!]
\centering
\includegraphics[width=0.5\textwidth]{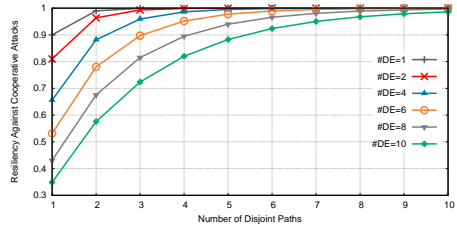}
\caption{The resiliency of multi-disjoint-path solutions against node capture.}
\label{Fig::reliability}
\end{figure}

To find the set of vertex-disjoint paths between any pair of source and destination vertices, we use the Ford-Fulkerson max-flow algorithm \cite{flow,fordFulkerson}. We know that the upper bound of the number of vertex-disjoint paths is $k$, because each node stores just $k$ keys, i.e. the source node has only $k$ neighbors in the overlay. The Ford-Fulkerson max-flow algorithm is known as a greedy algorithm capable of finding the set with the maximum number of edge-disjoint paths. However, our problem is to find the set of vertex-disjoint paths, not edge-disjoint. To find such a set, we modify this algorithm by replacing each vertex in our graph with two vertices which are connected with a directed edge, with a capacity of one.

\begin{lemma}
\label{maxFlow}
Consider directed graph $G(V,E)$, where $V$ and $E$ denote sets of vertices and edges, respectively, and the capacity of all edges is one. We modify $G$ to form a new graph $G'(V',E')$ by replacing each vertex $v_i$ with two vertices $v_{i1}$ and $v_{i2}$ and an edge from $v_{i1}$ to $v_{i2}$ with capacity one. Applying the Ford-Fulkerson algorithm on the modified graph results in a set with maximum number of vertex-disjoint paths in the main graph.
\end{lemma}

\begin{proof}
Assume, by contradiction, that the result of the Ford-Fulkerson algorithm on graph $G'$ does not return the maximum number of vertex-disjoint paths in graph $G$. This implies that there is at least a flow in graph $G'$ which passes through node $v_{i1}$ and then another node $v_{i2}'$ instead of passing $v_{i2}$. This, however, contradicts our assumption about $G'$ since in $G'$, there exists only a single edge with capacity one from each node $v_{i1}$ to $v_{i2}$. This is a contradiction and hence the results of the Ford-Fulkerson algorithm on the graph $G'$ returns the maximum number of vertex-disjoint paths in $G$.
\end{proof}

Fig. (\ref{disjointPaths}) shows the average number of disjoint paths and the distribution of their length. Fig. (\ref{disjointPaths}a) that includes the results for different numbers of nodes and different $k$ values shows two important facts. First, the number of vertex-disjoint paths is very close to the value of $k$. Second, increasing the number of nodes has a negligible impact on the number of vertex-disjoint paths. Collectively, Fig. (\ref{Fig::reliability}) and Fig. (\ref{disjointPaths}a) indicate that, with high probability, there will be enough number of vertex-disjoint paths for KPsec's operations. 

Fig. (\ref{disjointPaths}b) shows the distribution function of disjoint path length for different $k$ values in a graph with $1000$ nodes. This parameter is of paramount importance for the KPsec algorithm as a performance as well as a security metric. Although the encrypted data in the proposed algorithm follows the shortest physical path toward the destination, longer key-path length for vertex-disjoint paths leads to more network controlling traffic during the key-exchange process. Plus, longer key-paths mean more intermediate D-E steps and more vulnerability against cooperative attacks. 
Fig. (\ref{disjointPaths}b) shows that the length of the most vertex-disjoint paths is very close to the minimum key-path length reported in \cite{waina} and increasing the value of $k$ decreases the average key-path length and its variance which implies that the length of most disjoint key-paths is close to the average length. While not reported here, we investigate the same scenario for a fix $k$ value and different numbers of network nodes. The results are similar to those of Fig. (\ref{disjointPaths}b).

\begin{figure*}[t!]
\centering
\subfloat[Average number of disjoint paths.]{\includegraphics[width=0.5\textwidth]{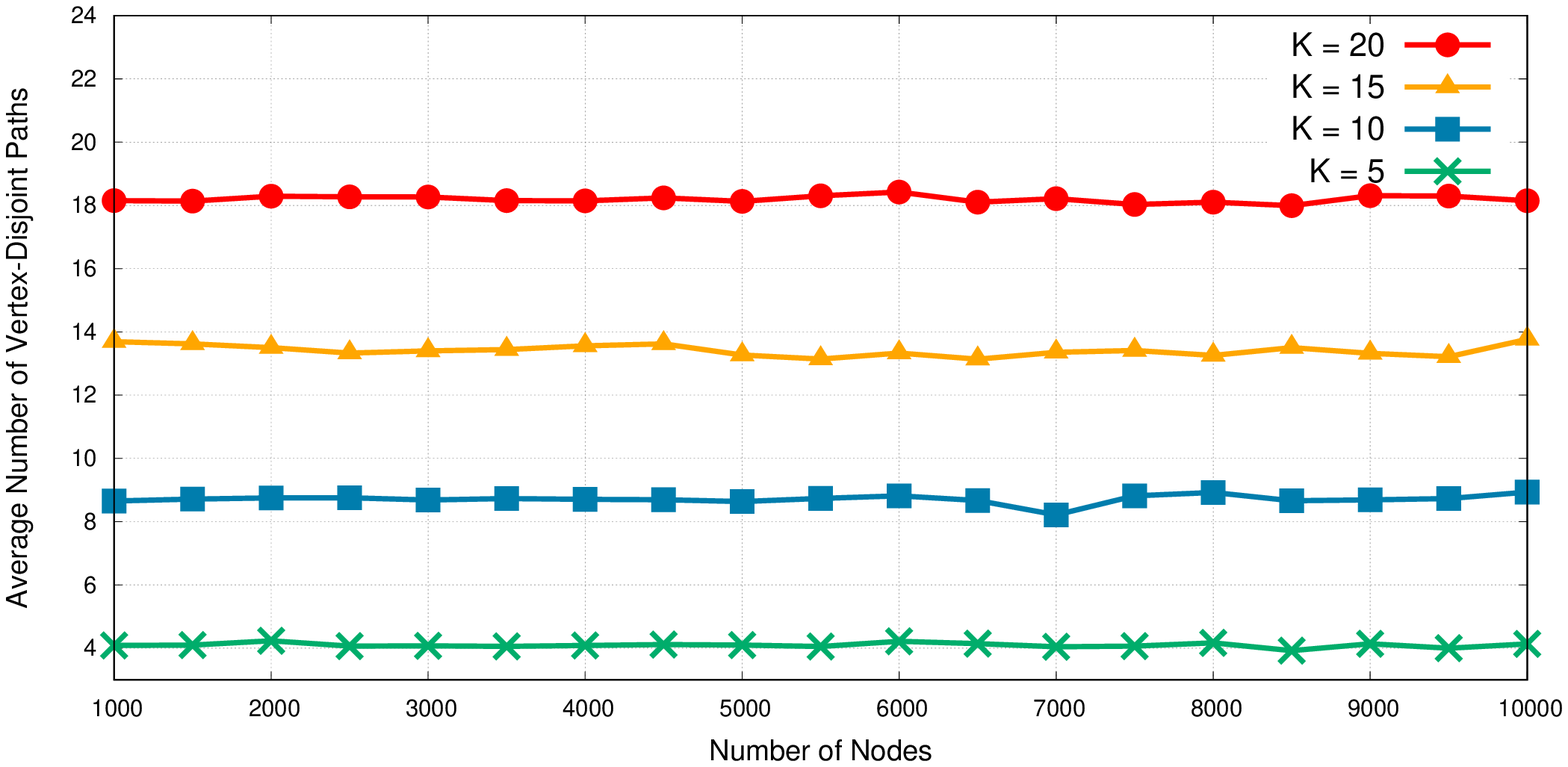}}
\subfloat[Distribution function of path length.]{ \includegraphics[width=0.5\textwidth]{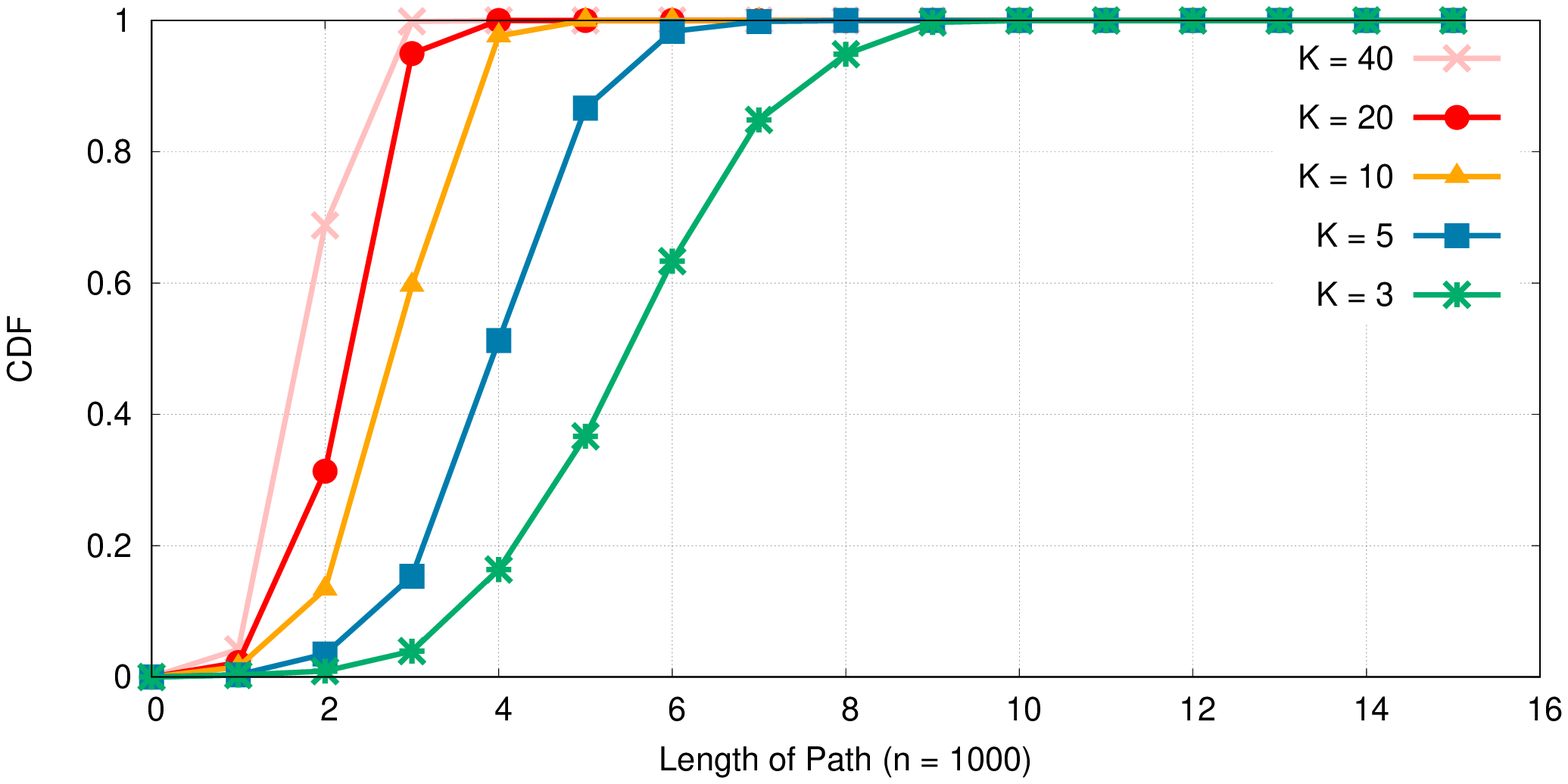}}
\caption{Vertex-disjoint paths in KPsec algorithm.}
\label{disjointPaths}
\end{figure*}

\subsection{Resiliency Against Cooperative Attacks}

In this part, we investigate the resiliency of the proposed algorithm against the cooperative attacks. In our analysis, we consider the resiliency against the cooperative man-in-the-middle attack, as one of the most known harmful attacks against multi-path solutions. However, our method can be generalized to any cooperative attacks. To model this attack, we introduce adversary nodes to the network and then calculate the number of those vertex-disjoint paths that do not contain any adversary node. We select the adversary nodes uniformly at random in our simulations. Fig. (\ref{securePaths}) shows the average and standard deviation of the number of secure vertex-disjoint paths.

\begin{figure*}[t!]
\centering
\subfloat[Average number of secure paths.]{\includegraphics[width=.5\linewidth]{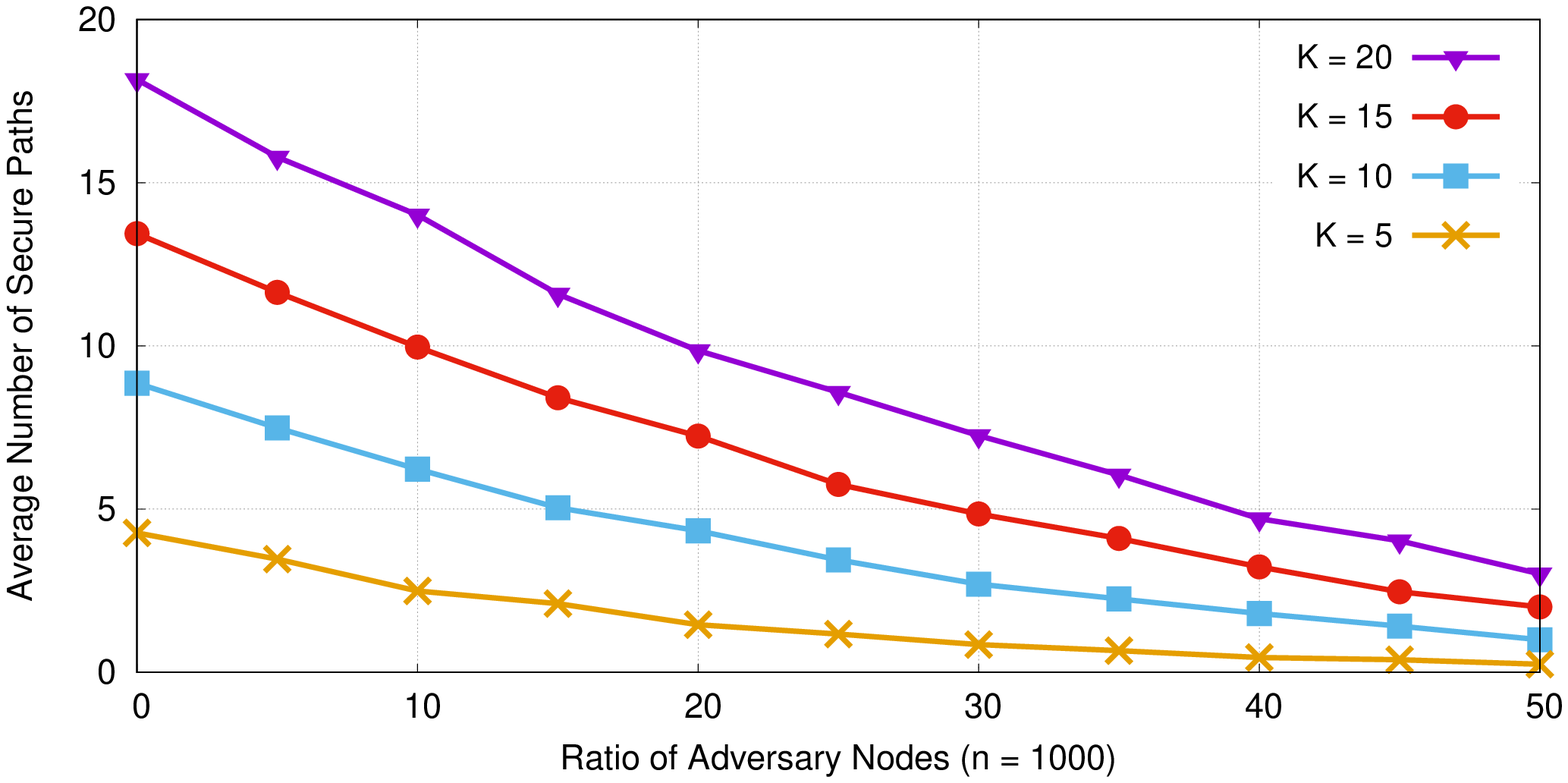}}
\subfloat[The standard deviation of the number of secure paths.]{ \includegraphics[width=.5\linewidth]{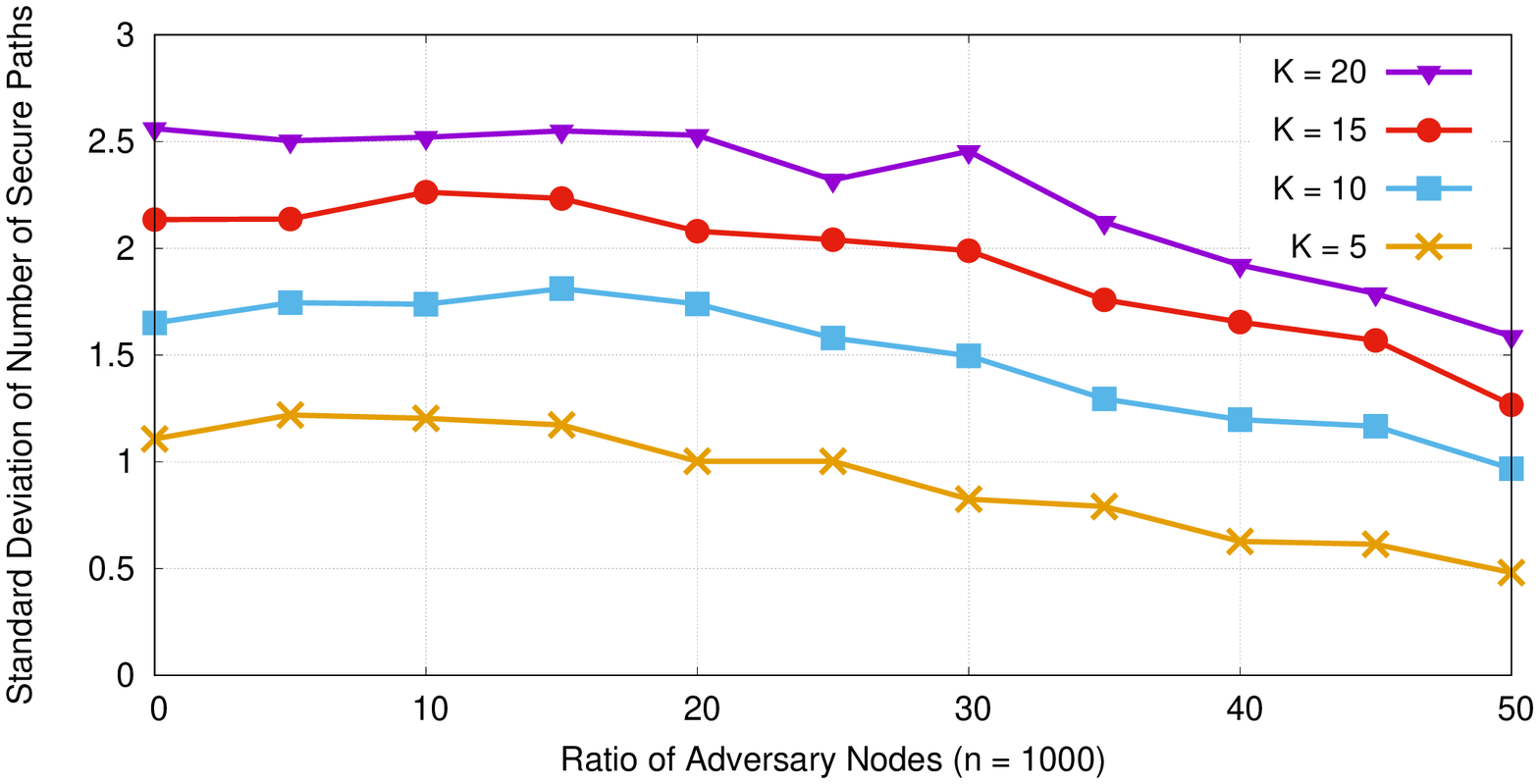}}
\caption{Secure vertex-disjoint paths for different ratios of adversary nodes.}
\label{securePaths}
\end{figure*}

Note that this parameter has to be analyzed together with parameter $\theta$. Recall that $\theta$ represents the number of required shares to rebuild the source node's public key. Let $\theta=\rho$, i.e., the destination requires all the shares from all the paths to become able to reconstruct the key. Thus, for the attacker to successfully perform its cooperative man-in-the-middle attack, it needs to compromise at least one node from every single path. According to Fig. (\ref{securePaths}), the attacker needs to compromise more than half of the network nodes to become successful. Decreasing the value of $\theta$ makes the system more resilient to failures, but it increases the probability of successful attacks. Nevertheless, even for $\theta=\frac{\rho}{2}$, the attacker needs to compromise more than $20\%$ of nodes to perform a successful attack. Considering the results of Fig. (\ref{securePaths}) for different values of $k$, we can conclude that for $\theta= c_1 \rho$, the attacker always needs to compromise $c_2 n$ nodes to perform a successful attack, where $c_1$ and $c_2$ are scaling constants, i.e. $0<c_1,c_2<1$. Hence, the attacker needs to compromise $O(n)$ nodes, even for small $\theta$ values.

\section{Experimental Testbed and Simulation Results}
\label{sec:experimental}

In order to evaluate the performance and security of KPsec in real systems and at scale, we implement it in a 10-node testbed as well as a large-scale ns-2 simulator\cite{ns2}. We implement KPsec and three state-of-the-art key pre-distribution schemes, PAKP \cite{waina}, SST \cite{infocom2011}, and $2$-UKP \cite{TWC2013}. 
We select a combination of both symmetric and asymmetric schemes to make a fair comparison. As we mentioned in \S 3, PAKP \cite{waina} is an asymmetric key pre-distribution scheme with a high probability of connectivity and a logarithmic number of D-E steps. However, it suffers from high energy consumption. Both of 2-UKP \cite{TWC2013} and SST\cite{infocom2011} are symmetric key pre-distribution schemes. 2-UKP has a high key sharing probability and consequently a shorter key-path. However, compromising a few numbers of nodes in this scheme leads to the compromise of many connections. In contrast, SST has a low key sharing probability that does not exceed $0.25$ which means longer key-path and consequently lower performance.

To make the connections of SST and UKP end-to-end secure, we augment these schemes with the algorithm of \cite{globcom05}. For performance parameters, we measure the average throughput, the overall network routing traffic, the key-exchange routing traffic overhead, the end-to-end latency, the key-exchange delay, and the energy consumed for decryption and encryption. For security metrics, we measure the number of intermediate D-E steps, the resiliency against cooperative attacks, the resiliency against passive attacks, and the resiliency against selective node compromise. The rest of this section is divided into four parts discussing testbed experiments, simulation settings, performance evaluation, and the comparison of the security strengths of different techniques.

\subsection{Experimental Testbed}

In our 10-node testbed experiment, each node stores 3 keys where two disjoint paths are used for the key exchange process. We used 10 laptops to perform the experiment by connecting them in an ad-hoc mode via a 5 Megahertz (MHz) wireless channel, 2.412-2.417 GHz. In each scenario, a 5 Megabytes (MB) file is sent from a specific source node to a specific destination. To make a fair comparison, we considered the same physical arrangement for all scenarios. We measured the time of the key-exchange process and the time between sending the first data packet by the source and receiving the last packet by the destination. The overall end-to-end latency is the summation of these times. We further measured the control traffic required for the key-exchange process in each algorithm. The throughput is also measured as the packet delivery ratio over the bandwidth. Table (\ref{tbl::testbed}) shows the result of our testbed experiments.

\begin{table*}[ht]
\caption[Table caption text]{Experimental testbed results.}
\resizebox{1\textwidth}{!}
{\begin{minipage}{\textwidth}
\begin{center}
\begin{tabular}{ | l | l | l | l | l | l | l | }
\hline
&Key-Exch.&Data&End-to-End&Key-Exch.& Throughput& D-E\\
Algorithm&Delay &Transmission&Latency&Traffic&(Mbps)& Steps\\
&(Sec)&Latency(Sec) &(Sec)&(KB)&&(per path)\\ \hline \hline
Aug. UKP& 13.98&214.57&228.55&436.85&0.9873&1\\
Aug. SST& 48.83&274.60&323.43&948.84&0.9579&3.5\\
KPsec& 16.06&209.92&225.98&868.88&0.999&1\\
\hline
\end{tabular}
\end{center}
\label{tbl::testbed}
\end{minipage} }
\end{table*}

Due to the low key-sharing probability of SST (discussed in \S\ref{sec:relatedWork}), the key-path in this scheme is significantly longer than other schemes. This fact leads to longer key exchange delay and higher key exchange traffic. Since, in all cases, the data follows the shortest physical path to reach the destination, the data transmission latency and throughput are expected to yield similar results. However, due to the network traffic and latency caused by the key-exchange process, we observe lower throughput for the augmented SST algorithm. 

\subsection{Simulation Setting}

To evaluate the algorithm at scale, we use the ns-2 simulator. In each scenario, a network with a number of nodes (ranging from $100$ to $200$) is simulated in a $300 \times 300$ square meter area. Network nodes are initiated in random positions, using a uniform distribution in the network area. 
The nodes are assumed to be mobile and follow Random walk mobility model of ns-2, with zero pause time and varying speed in the interval $[0\quad 5]$ meter per second. The distance model is chosen for sending and receiving with the communication range of $100$ meters for each node. The channel bandwidth is set to $1Mbps$. 
All simulations are performed using AODV routing protocol to find the shortest physical path. For two-layer routing in PAKP, the algorithm of \cite{TNSE} is used to find the optimal path with the smallest number of D-E steps and shortest physical lengths. Different scenarios are simulated with different numbers of connections between $10$ and $20$. To keep the comparisons fair, all connections are chosen randomly but once selected, the same connections are used for comparing different schemes. The generated traffic is FTP running on TCP Tahoe. In each connection, the source node sends a file with a size of $5$MB to its destination. All simulations are repeated $20$ times, and figures show average values calculated over all runs. For the key pre-distribution phase, the keyring size is set to $k=10$, and for end-to-end algorithms, we use five disjoint paths in each scenario.

\subsection{Performance Evaluation}
\label{performanceEvaluation}

We choose the network throughput measured for successful packet delivery, the average end-to-end latency per connection, the average key-exchange delay, the average routing traffic per connection, the key-exchange routing traffic, and the consumed energy as performance evaluation metrics.
Fig. (\ref{throughput}) compares the throughput of different scenarios. While Fig. (\ref{throughput}a) shows the average throughput for $10$ fixed connections and the different number of nodes, Fig. (\ref{throughput}b) shows the results for different numbers of connections in a $100$-nodes setting. Since increasing the number of connections increases congestion, the network throughput is slightly decreased as the number of connections increases. We observe that the factor that impacts the network throughput the most is the physical path length. Since, in end-to-end solutions, the data traffic follows the shortest physical path, KPsec's throughput is higher in comparison with SST, 2-UKP, and simple PAKP, as shown in Fig. (\ref{throughput}). It is worth noting that augmented SST and augmented UKP have similar throughput as KPsec. Since $2$-UKP has a significantly higher number of overlay edges, it has the shortest physical path among the compared schemes. This fact leads to $2$-UKP outperforming other key pre-distribution schemes. An improvement of more than $7.5\%$ is also notable for KPsec compared to PAKP. 

\begin{figure*}[t!]
\centering
\subfloat[Different number of nodes.]{\includegraphics[width=.5\linewidth]{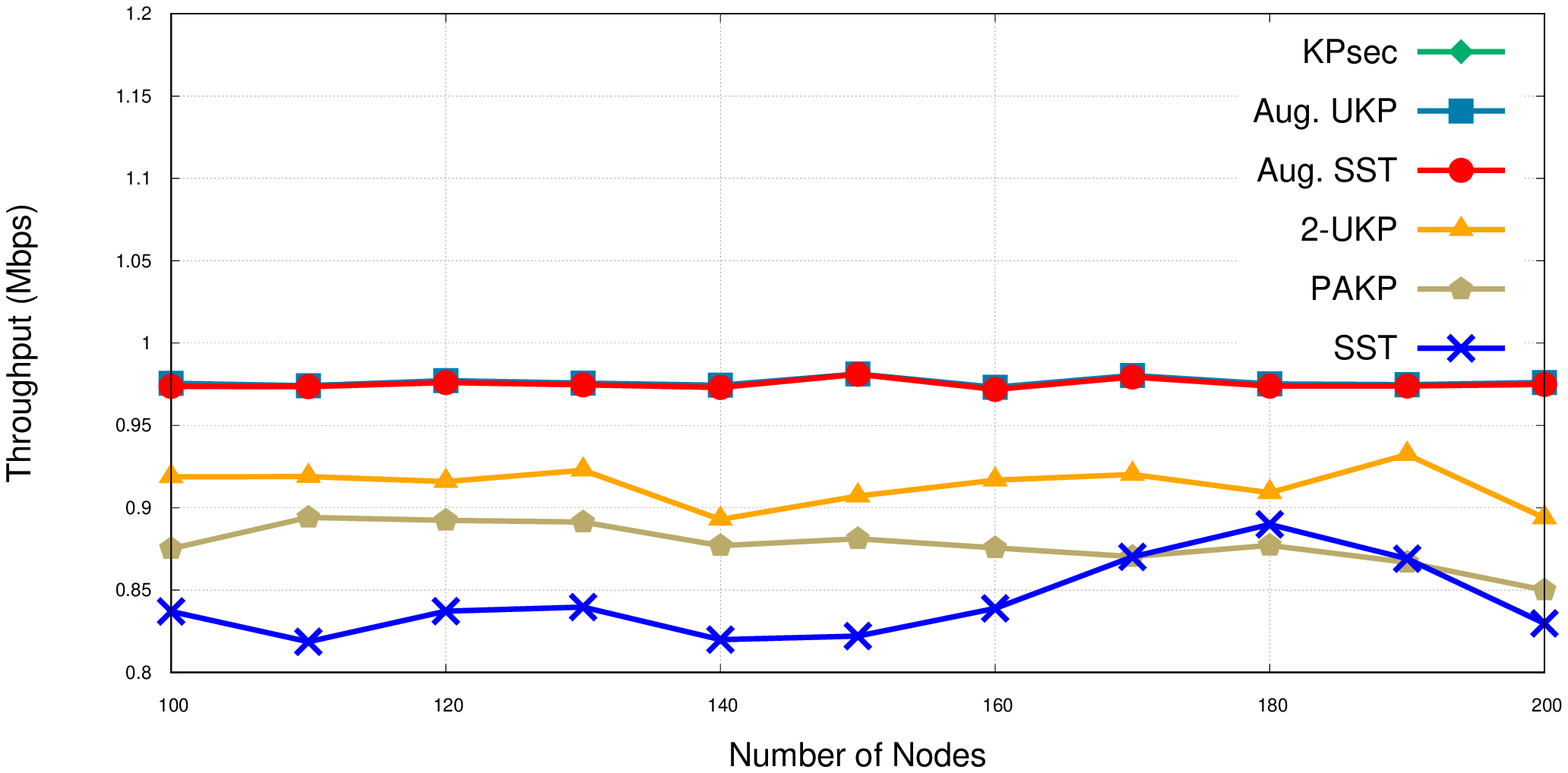}}
\subfloat[Different number of connections.]{ \includegraphics[width=.5\linewidth]{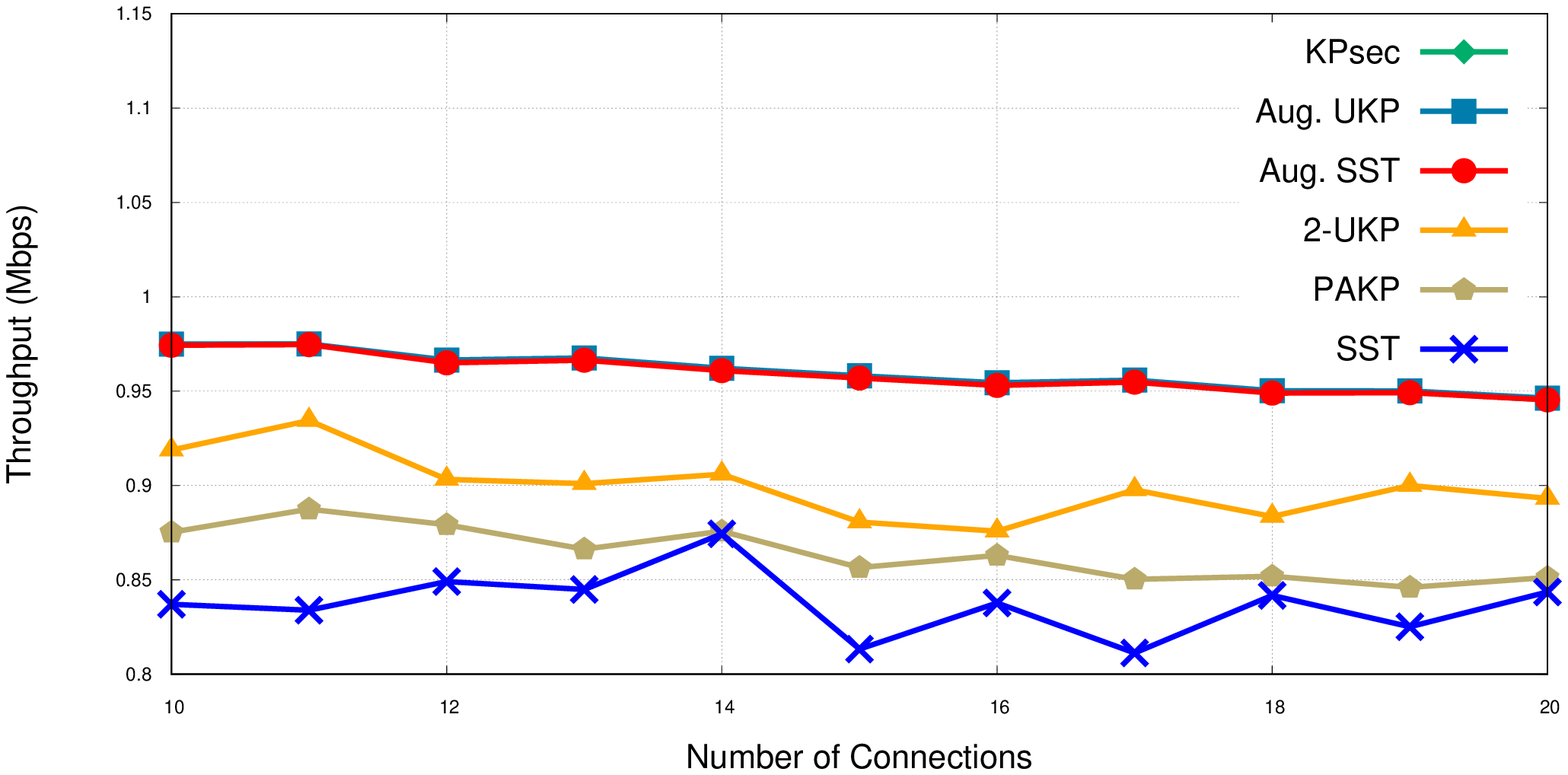}}
\caption{A comparison of the average network throughput.}
\label{throughput}
\end{figure*}

We next measure the average end-to-end latency per connection. Each connection starts at a time randomly chosen within the interval $[0,60]$ seconds. The end-to-end latency for each connection ends when the destination receives the last packet of the file. The average latency per connection is shown in Fig. (\ref{delay}) for different numbers of nodes and connections. Consistent with our testbed results, all the end-to-end solutions exhibit similar performance. KPsec shows significant improvement of more than $50\%$ compared to SST, $2$-UKP, and PAKP. While PAKP slightly improves the performance compared to SST, $2$-UKP outperforms both of them.

We also measure the key-exchange delay. Fig. (\ref{delay-keyExchange}) shows the results for different numbers of nodes and connections. Recall that the key-exchange process in KPsec has one additional step in comparison with the algorithm of \cite{globcom05}. In KPsec, after receiving key shares, the destination node encrypts its public key and sends it to the source node. This extra step imposes some delay which leads to augmented UKP outperforming KPsec for this metric. However, the longer key-path in SST increases the augmented SST key-exchange's delay.

\begin{figure*}[t!]
\centering
\subfloat[Different number of nodes.]{\includegraphics[width=.5\linewidth]{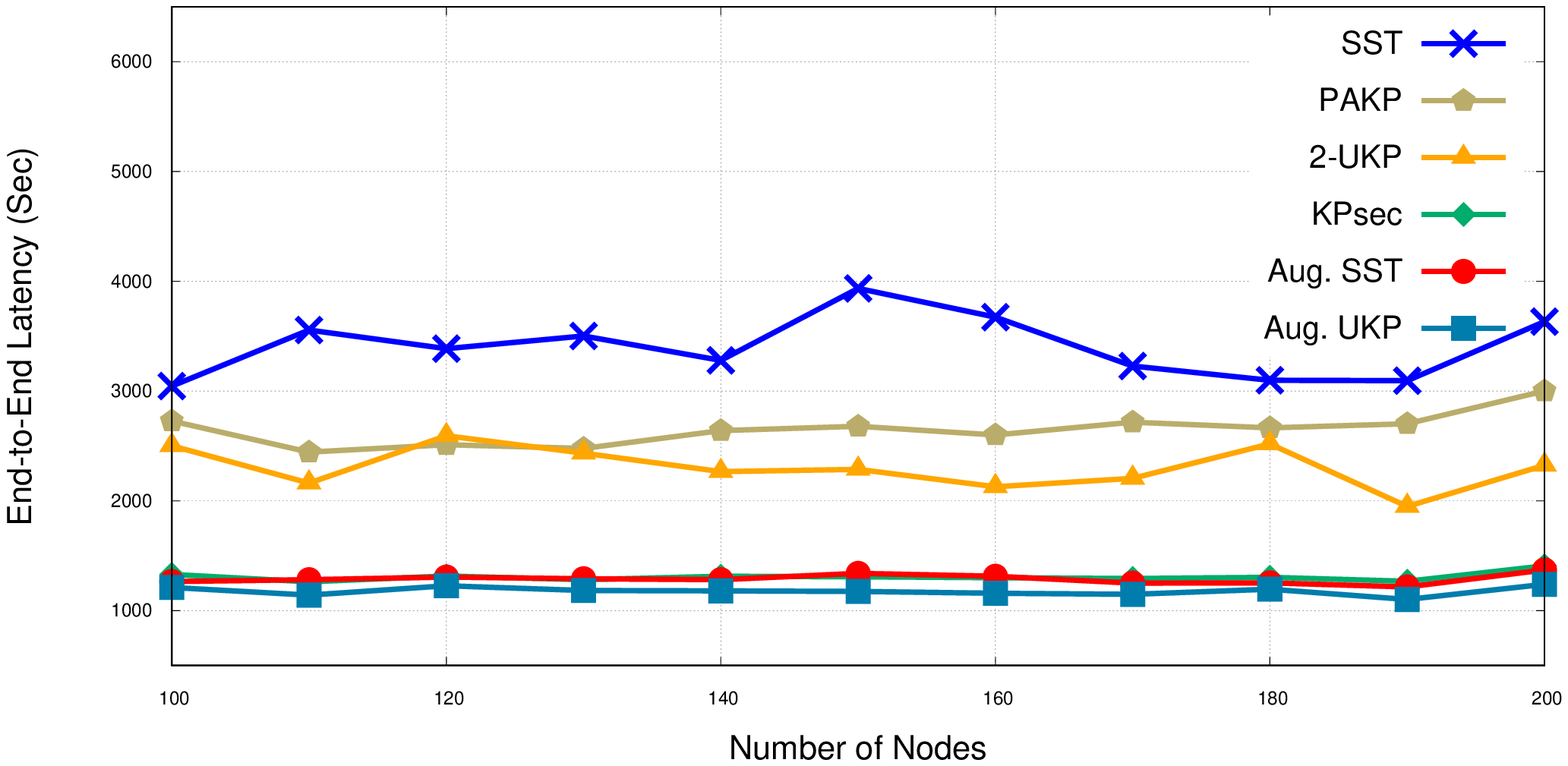}}
\subfloat[Different number of connections.]{\includegraphics[width=.5\linewidth]{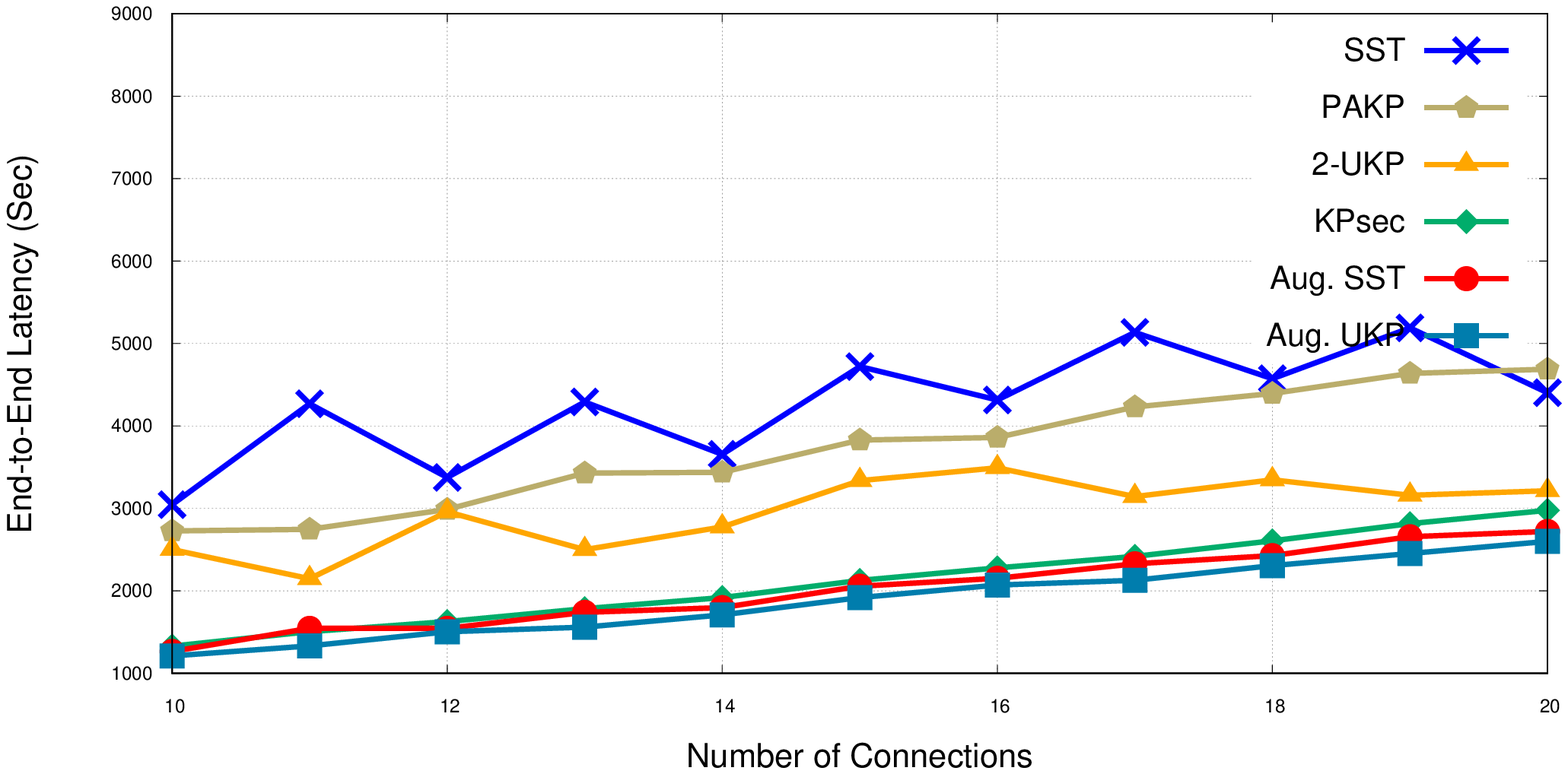}}
\caption{A comparison of the average end-to-end latency per connection.}
\label{delay}
\end{figure*}

\begin{figure*}[t!]
\centering
\subfloat[Different number of nodes.]{\includegraphics[width=.5\linewidth]{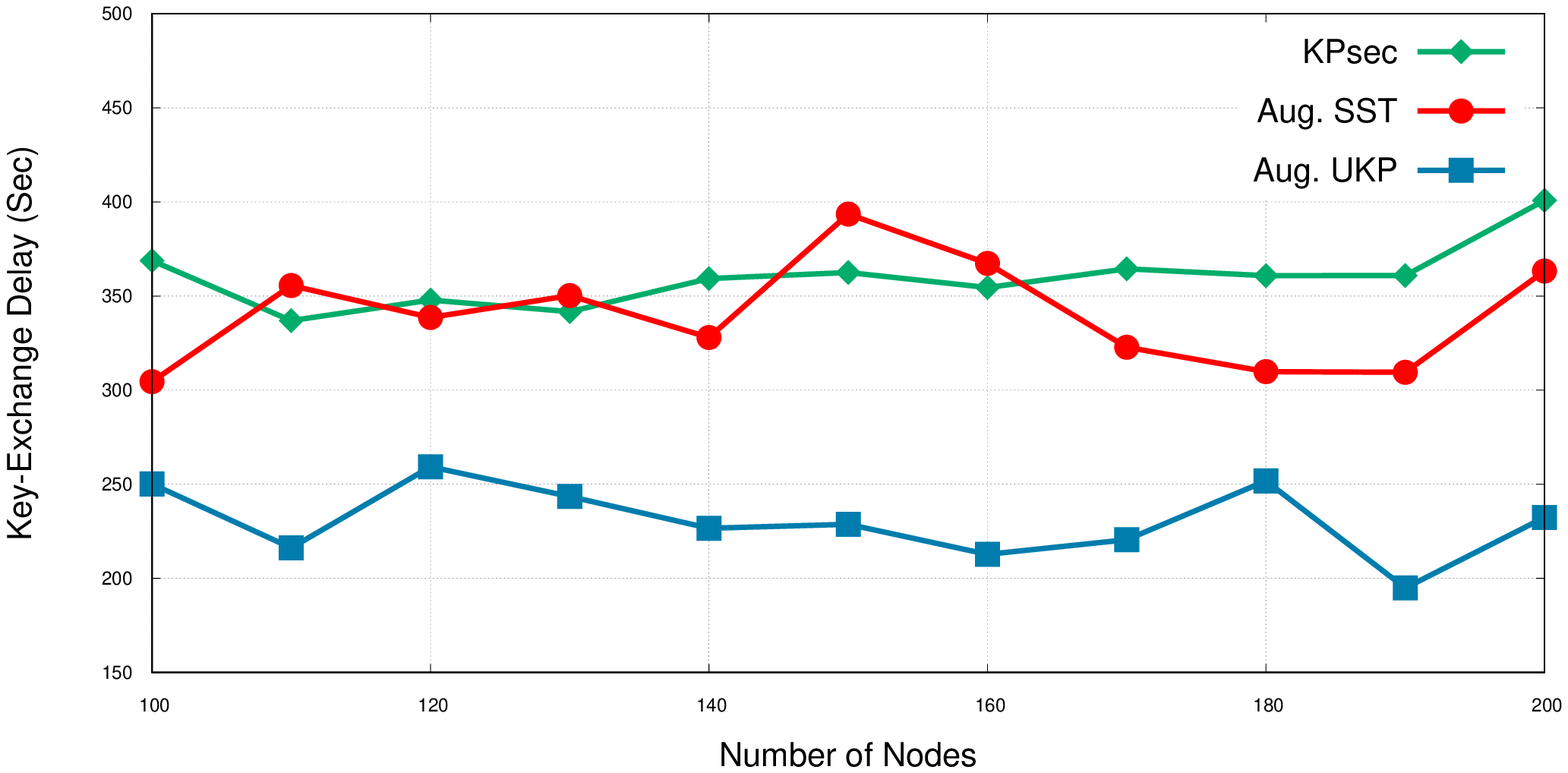}}
\subfloat[Different number of connections.]{\includegraphics[width=.5\linewidth]{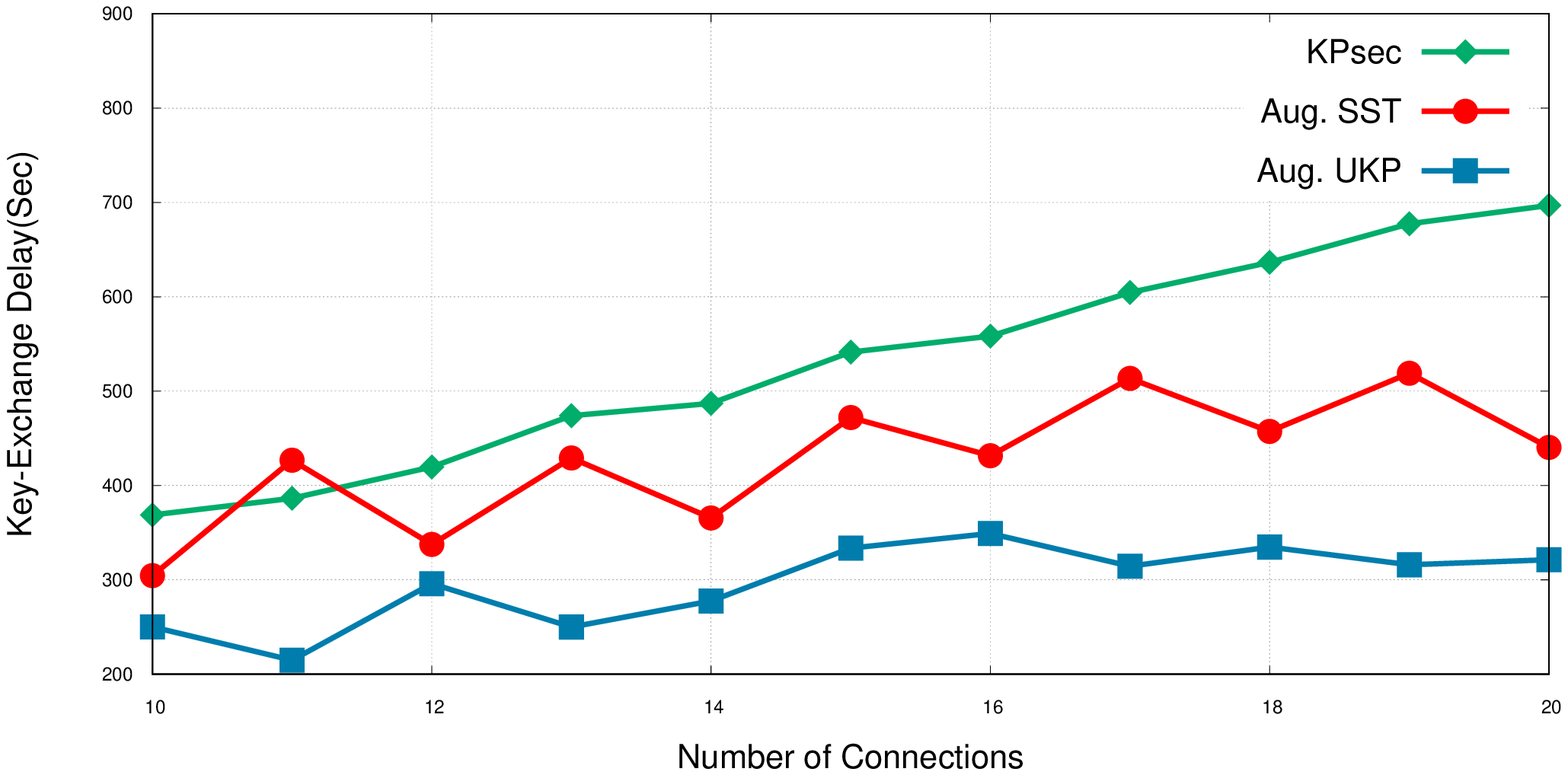}}
\caption{A comparison of the average key-exchange delay.}
\label{delay-keyExchange}
\end{figure*}

We next measure the routing traffic overhead generated for sending the encrypted $5$MB files. Fig. (\ref{routing}) shows this parameter measured in MB. Again, a longer physical path degrades this parameter for both SST and augmented SST. Fig. (\ref{routing-keyExchange}) shows the key-exchange traffic for end-to-end algorithms. This figure shows that KPsec and augmented UKP generate almost similar volumes of key-exchange traffic which is lower than augmented SST. While not shown here, a network with stationary nodes is also simulated. The results show the same pattern for all the mentioned parameters. However, for routing traffic, the network with stationary node shows an average of 9\% less overall routing traffic and 11.5\% less key-exchange traffic overhead. We have used the setting of \cite{energy} to calculate the consumed energy for encryption and decryption processes in our simulations. Fig. (\ref{energy}) shows the results for a network with different numbers of nodes. Since PAKP encrypts data asymmetrically, it consumes an order of magnitude more energy in comparison with other algorithms. Thus, we remove its curve for better representation. The SST scheme has more intermediate D-E steps in comparison with $2$-UKP. Thus, it consumes energy at a rate almost twice as large as that of $2$-UKP. KPsec, in turn, outperforms the key pre-distribution schemes by more than $70\%$. Since the data transmission process in all end-to-end algorithms follows the shortest physical path and all of them use symmetric encryption, their performance with respect to this metric is similar.

\begin{figure*}[t!]
\centering
\subfloat[Different number of nodes.]{\includegraphics[width=.5\linewidth]{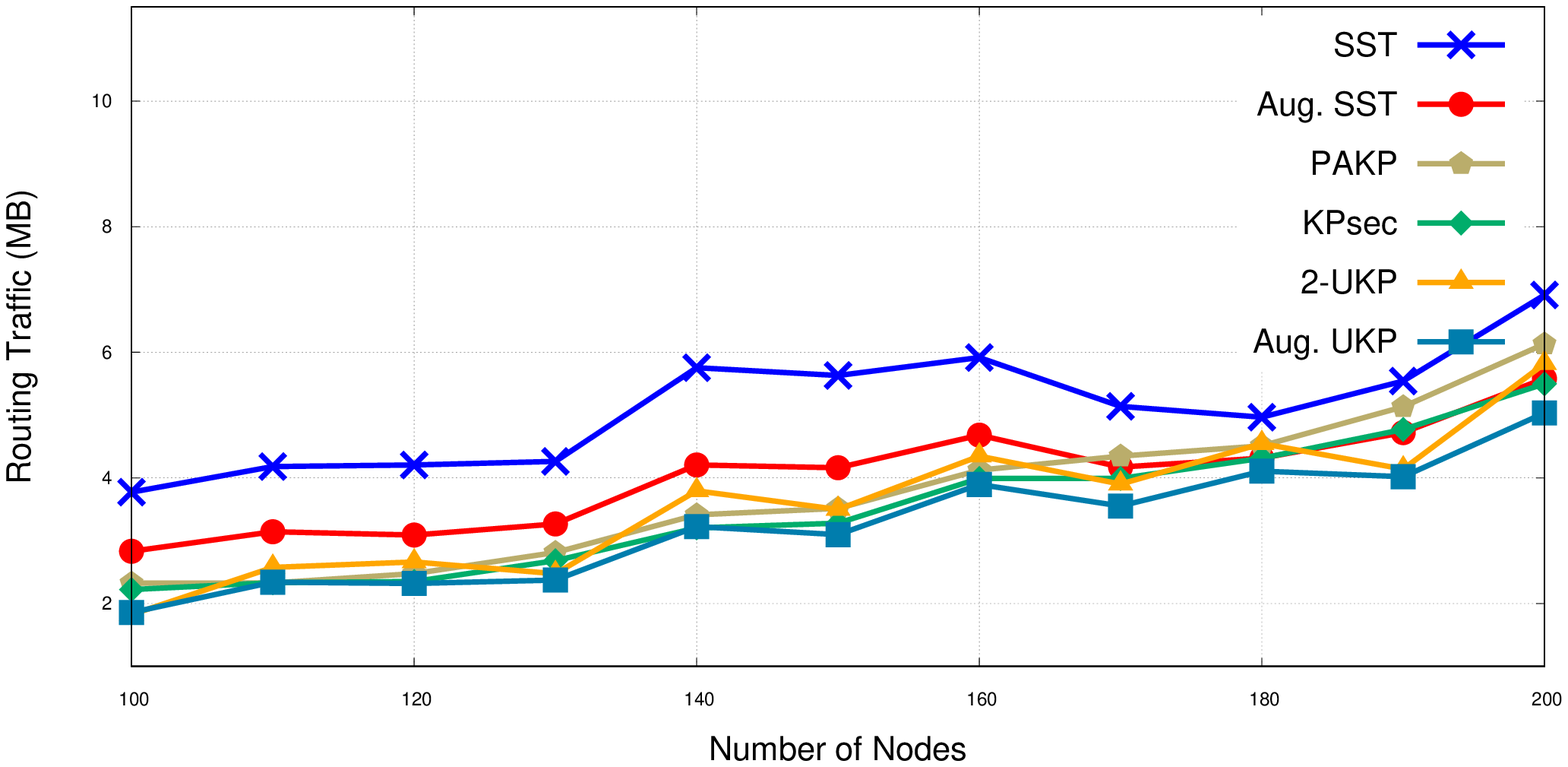}}
\subfloat[Different number of connections.]{ \includegraphics[width=.5\linewidth]{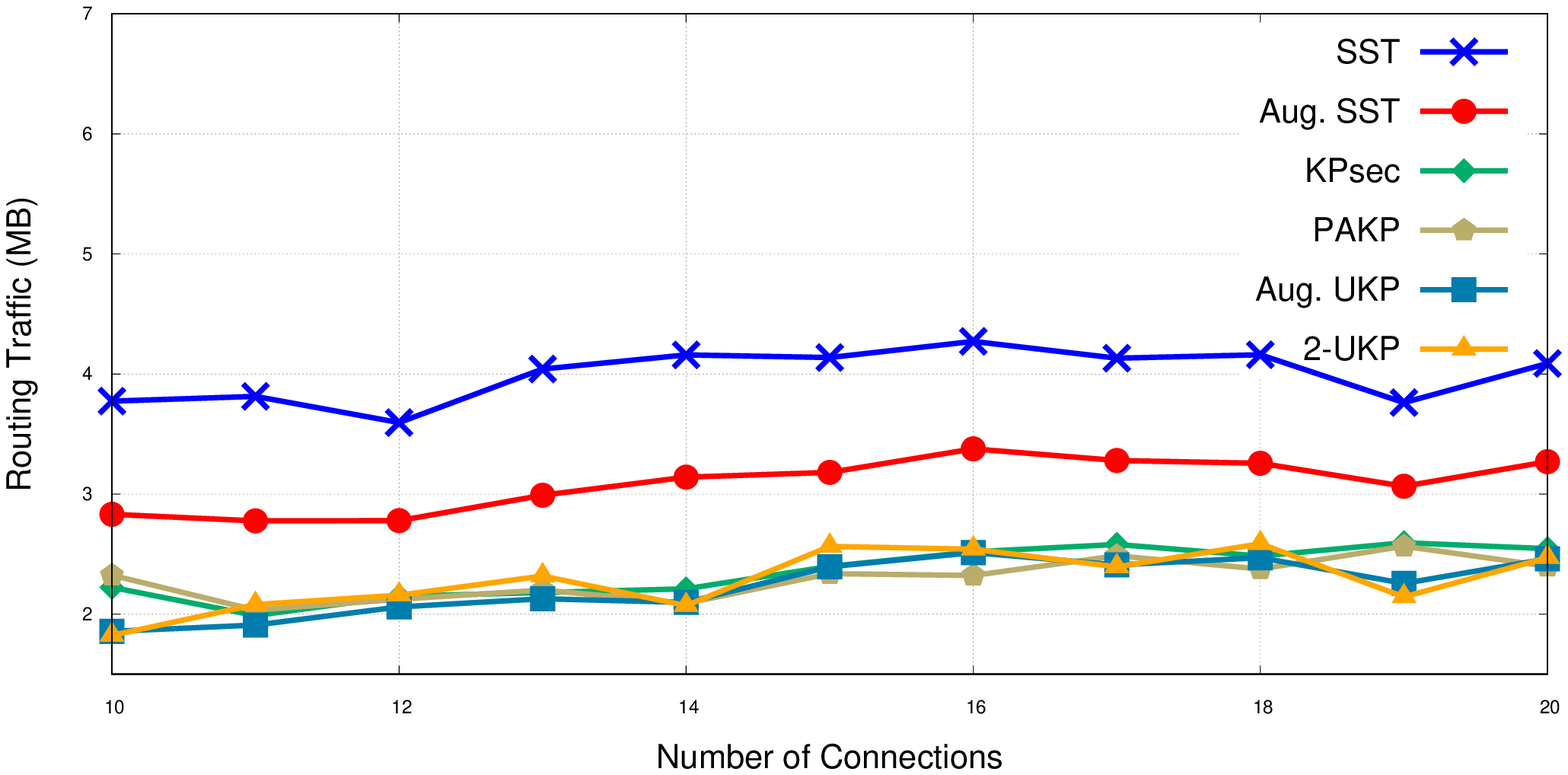}}
\caption{A comparison of the average network routing traffic.}
\label{routing}
\end{figure*}

\begin{figure*}[t!]
\centering
\subfloat[Different number of nodes.]{\includegraphics[width=.5\linewidth]{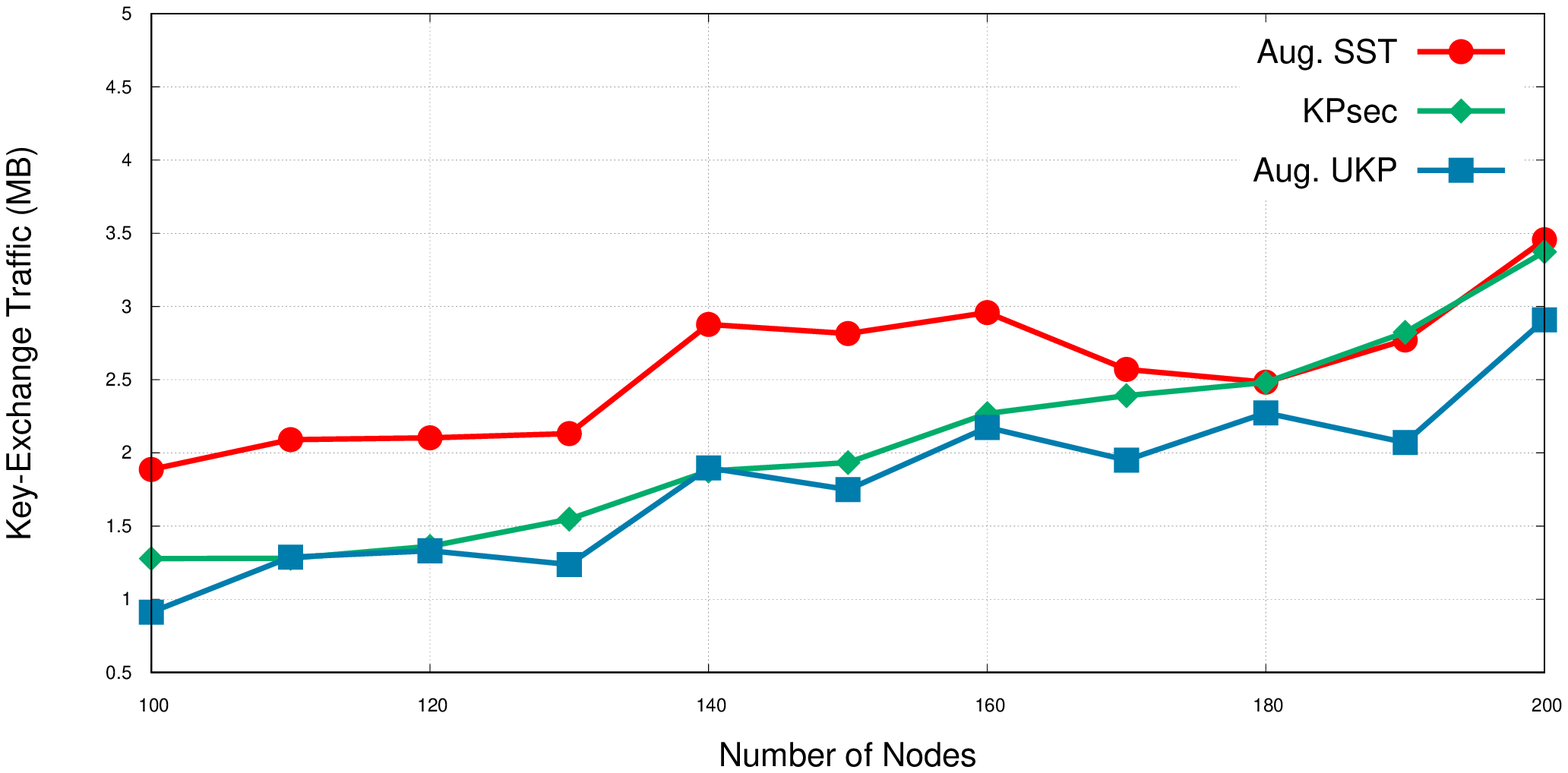}}
\subfloat[Different number of connections.]{ \includegraphics[width=.5\linewidth]{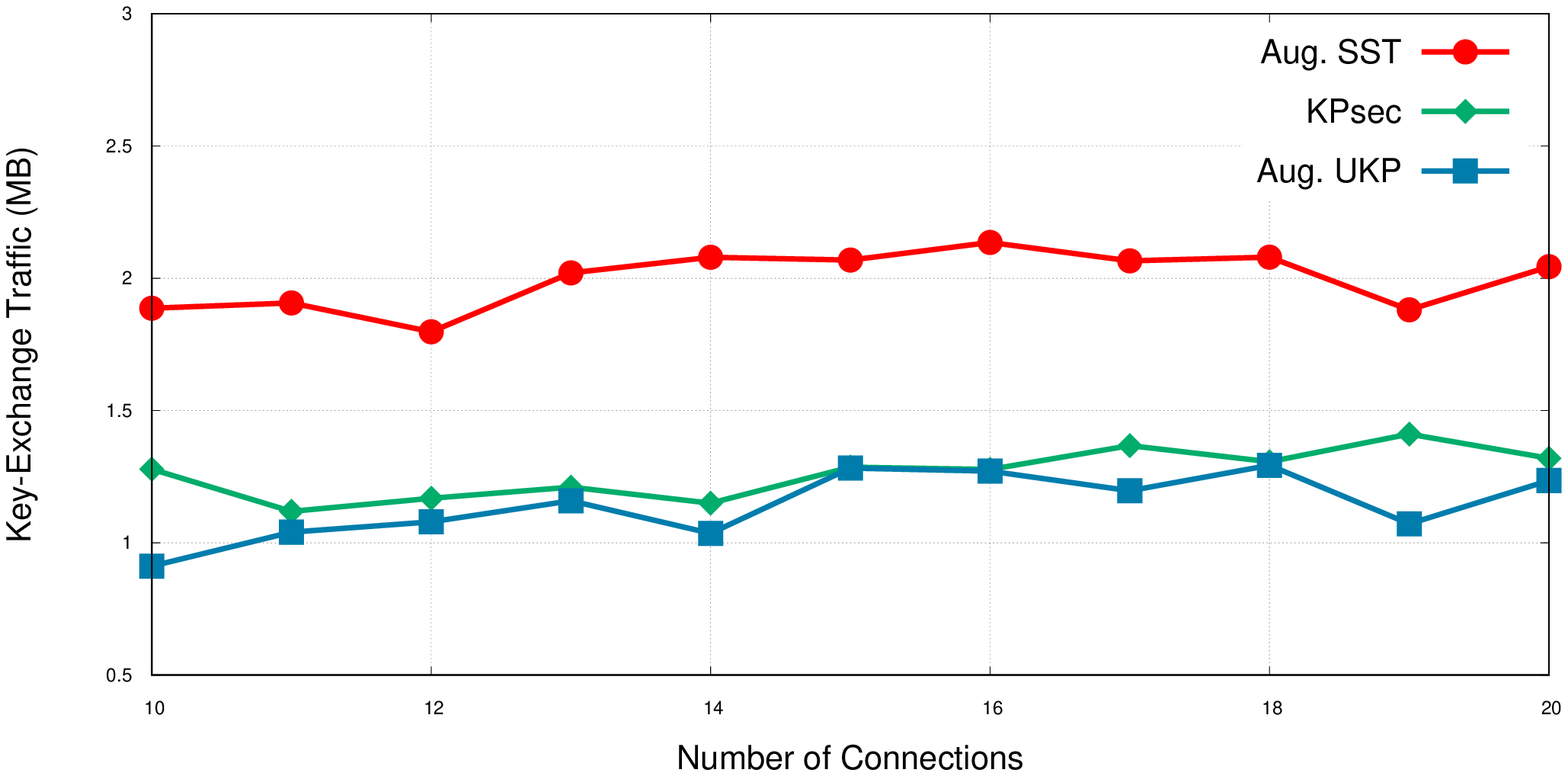}}
\caption{A comparison of the average key-exchange traffic.}
\label{routing-keyExchange}
\end{figure*}

\begin{figure*}
\centering
\includegraphics[width=0.5\linewidth]{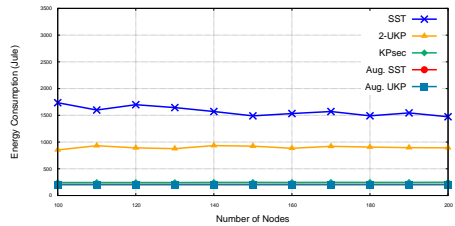}
\caption{Average energy consumed for encryption and decryption processes.}
\label{energy}
\end{figure*}

Overall, our results show that despite the fact that end-to-end solutions including KPsec add some delay and traffic overhead, they remove the path stretch and consequently result in better overall performance. They also show that, while the performance of \cite{globcom05} depends on its underlying key pre-distribution scheme, generally it is close to KPsec, performance-wise. However, in the following part, we show that KPsec has significant security advantages over \cite{globcom05}.

\subsection{KPsec Improves Security}
\label{securityComparison}

We first measure the average intermediate D-E steps in each disjoint path, as a basic security metric. Fig. (\ref{DEsteps}a) shows this parameter for different schemes. This figure shows that the number of intermediate D-E steps in the KPsec algorithm is significantly lower than those of augmented SST and augmented UKP. While we represented a general analysis for resiliency against cooperative attacks in Fig. (\ref{Fig::reliability}), we combine the results of Fig. (\ref{DEsteps}a) with the mentioned analysis to show the resiliency of different algorithms. Fig. (\ref{DEsteps}b) shows the results for a network with 100 nodes, $10\%$ of them being compromised, and different numbers of disjoint paths. As Fig. (\ref{DEsteps}b) shows, KPsec approaches to perfect resiliency with only three disjoint paths, while this number is 5 and 8 for augmented UKP and augmented SST, respectively. That is, KPsec can use a lower number of disjoint paths to achieve higher performance for the same level of resiliency against cooperative attacks.

\begin{figure*}[t!]
\centering
\subfloat[Average number of D-E steps.]{\includegraphics[width=0.5\linewidth]{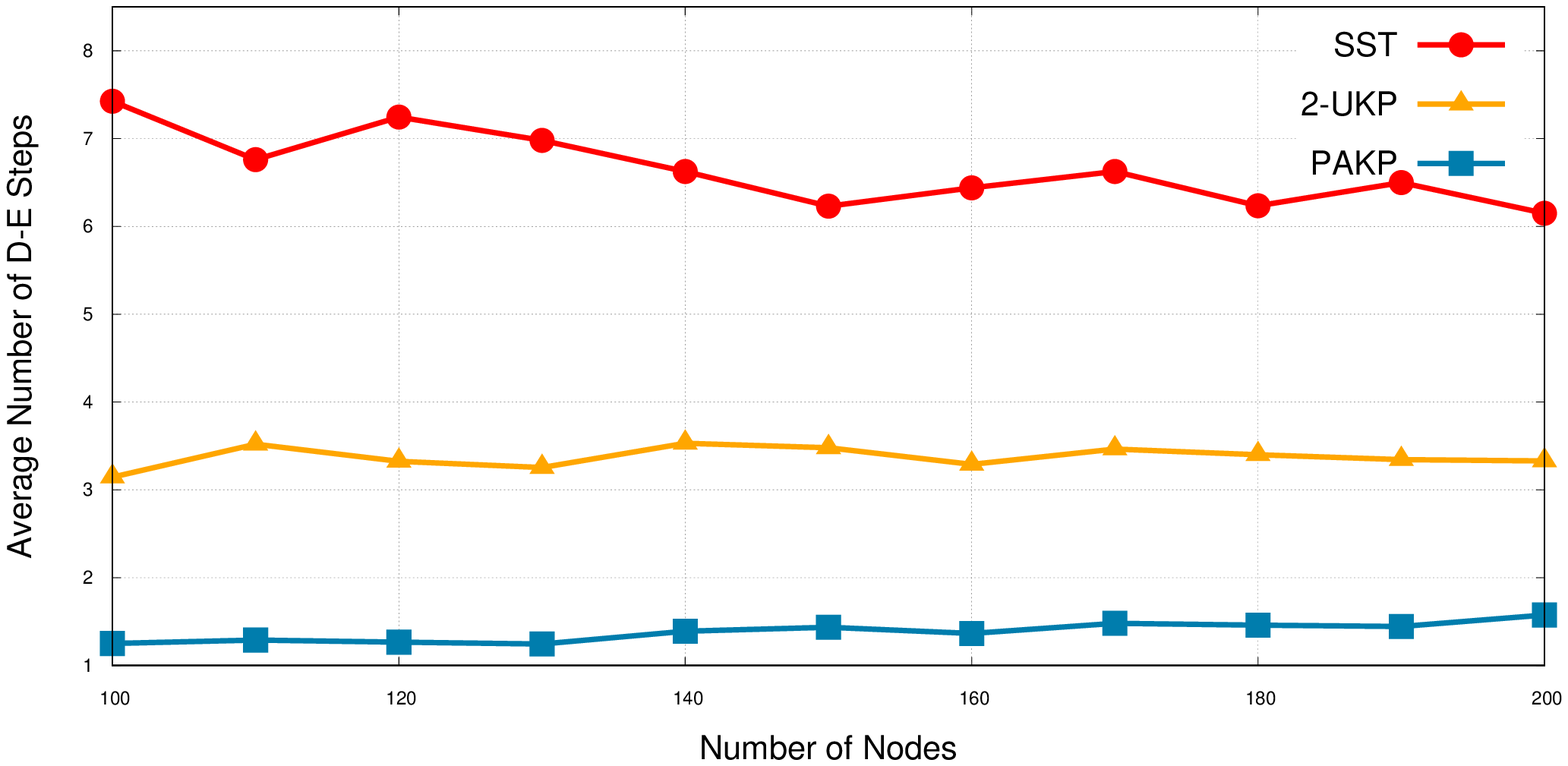}}
\subfloat[Resiliency against cooperative attacks.]{\includegraphics[width=0.5\linewidth]{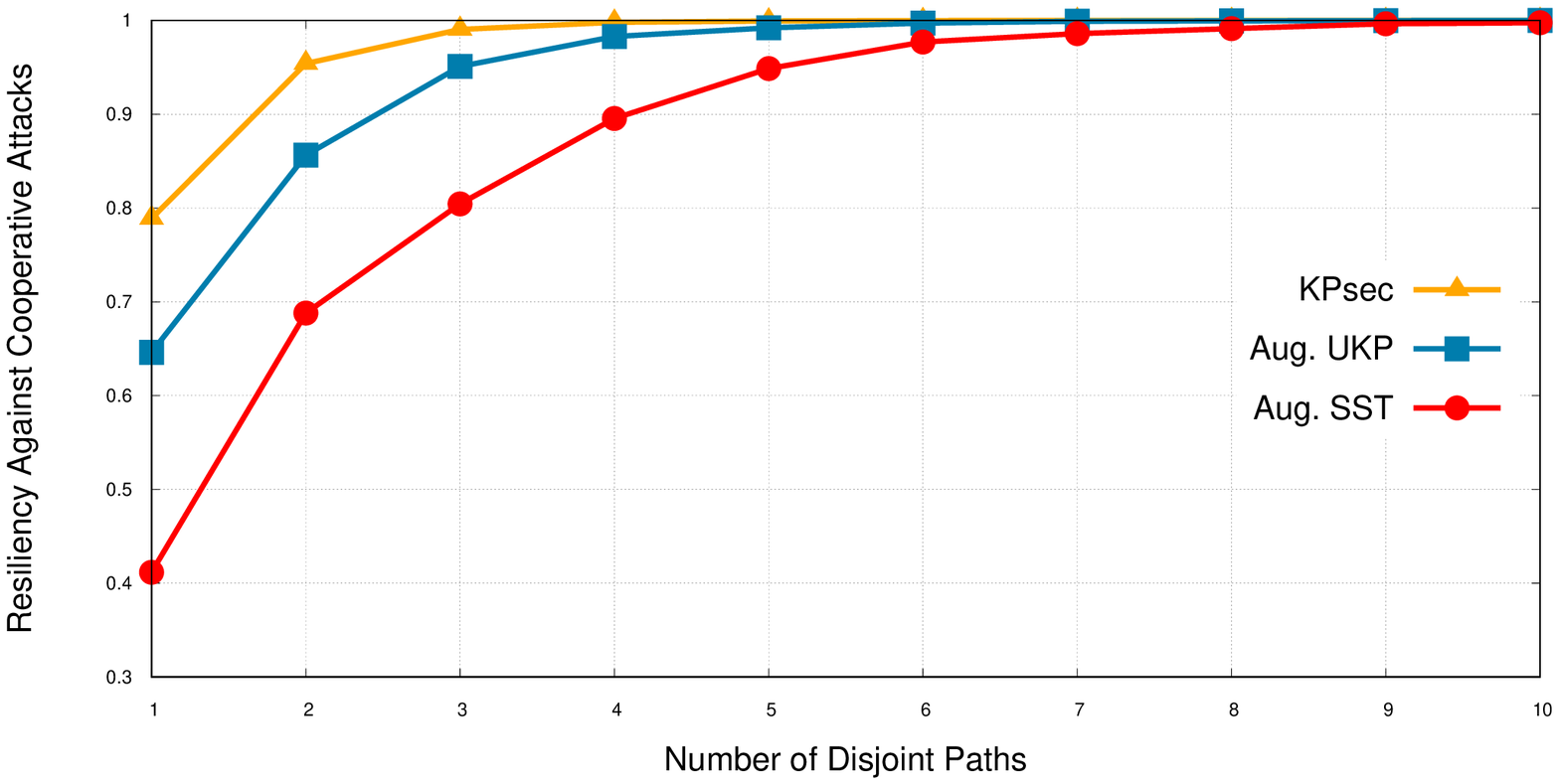}}
\caption{Average number of D-E steps and its effect on the network resiliency.}
\label{DEsteps}
\end{figure*}

\begin{table*}[ht]
\caption[Table caption text]{A brief security comparison.}
\resizebox{1\textwidth}{!}
{\begin{minipage}{\textwidth}
\centering
\begin{tabular}{ | l | l | l | l | l | l | l | }
\hline
& Encryption &Average&Cooperative&\#Nodes&Secure&Subject \\
Algorithm&System& \#D-E&Attack&for succ. &Inform.&to Passive \\
&&(per path)&Resiliency&SNC &Leakage&Attack\\ \hline \hline 
Aug. UKP&Symmetric&4.145&0.95&10&Positive&Positive\\
Aug. SST&Symmetric&8.425 &0.80&23&Positive&Positive\\
KPsec&Asymmetric&2.25 &0.99&99&Negative&Negative\\
\hline
\end{tabular}
\label{tbl::security}
\end{minipage} }
\end{table*}

One of the main advantages of KPsec is its resiliency against passive attacks such as eavesdropping. Since all transferred keys in KPsec are public, an attacker cannot degrade the computational hardness of the cryptosystem and consequently cannot compromise the secrecy of data transmission, by eavesdropping. In contrast, a large enough number of compromised nodes in augmented SST and augmented UKP enable the attacker to access the pairwise key only by eavesdropping. Even if the attacker cannot eavesdrop all the key pieces in the algorithm of \cite{globcom05}, it can access some key pieces and generate other parts by a brute-force search. By considering the fact that the computational hardness of symmetric cryptosystems exponentially increases by the increment of key length \cite{keyManagNIST}, knowing any portion of the key is equivalent to the shorter key length, and hence, it logarithmically decreases the computational complexity of the brute-force attack \cite{bruteForce}. In other words, the algorithm of \cite{globcom05} suffers from secure information leakage. The next advantage of KPsec over symmetric end-to-end solutions is the geographical distance of its overlay neighbors. In symmetric solutions, the attacker can perform a jamming attack and force the source node to establish its connection through a specific neighbor that the attacker desires (i.e., a compromised node). In KPsec, by contrast, since the overlay neighbors are, in most cases, physically far away and the physical neighbors carry only encrypted messages, this attack becomes ineffective. 

The next important security metric is the number of nodes that an attacker needs to capture in order to compromise the security of the network as a whole. This metric is sometimes referred to as the resiliency against selective node capture (SNC) attacks. In symmetric key pre-distribution schemes, the key-pool includes a limited number of secret keys. Hence, if the attacker knows about the keyring arrangement, it can selectively capture nodes to get access to the entire key-pool. In $2$-UKP and SST, this number is $O(\sqrt n)$ and $O(k)$ nodes, respectively \cite{optimalRouting}. Under PAKP, by capturing each node, the attacker accesses only several public keys and only one private key. Hence, the attacker needs to capture $O(n)$ nodes to access all private keys. Table (\ref{tbl::security}) summarizes and compares the security of these schemes for a network with 100 nodes and three disjoint paths.

\section {Conclusion}
\label{conclusion}

In this paper, we propose KPsec to address two main shortcomings of existing key pre-distribution schemes: the intermediate D-E steps and path stretch. KPsec establishes a pairwise key and makes end-to-end connections secure by deploying a key-exchange process using overlay disjoint paths. We evaluate the performance and security of KPsec as well as three state-of-the-art key pre-distribution schemes using real testbed and large-scale simulations. Our results show improvements in network throughput, end-to-end latency, and energy consumption. This is because the overhead of deploying multiple overlay disjoint paths is negligible in comparison with the performance benefits gained by removing the path stretch. We show that KPsec requires fewer number of disjoint paths to achieve the same level of resiliency against cooperative attack compared to other multi-path solutions. Furthermore, contrary to other algorithms, KPsec is resilient against passive attacks and does not suffer from the secure information leakage. KPsec's main goal is to protect the confidentiality of communications. We leave the availability analysis for future work.

\nocite{*}
\bibliographystyle{splncs04}
\bibliography{References}

\end{document}